\documentclass[%
 reprint,
preprintnumbers,
nofootinbib,
 amsmath,amssymb,
 aps,
 prl,
]{revtex4-2}

\usepackage{graphicx}
\usepackage{dcolumn}
\usepackage{bm}

\usepackage{amsthm}
\usepackage{mathrsfs}
\usepackage{customcommands}
\usepackage{dsfont}
\usepackage{comment}
\usepackage[utf8]{inputenc}
\usepackage{calc}
\usepackage{accents}
\usepackage{amsfonts}

\usepackage{braket}
\usepackage[normalem]{ulem}
\usepackage{color}

\usepackage{mathtools}
\usepackage{tensor} 
\usepackage{tikz}
\usetikzlibrary{arrows, positioning, quotes, intersections}

\usepackage{thmtools, thm-restate}

\declaretheorem{theorem}

\newtheorem{corollary}{Corollary}
\newtheorem{conjecture}{Conjecture}

\newcommand{\Int}{\text{Int}}
\newcommand{\cl}{\text{cl}}
\newcommand{\edge}{\text{edge}}
\newcommand{\sset}{\subseteq}
\newcommand{\cM}{\mathcal{M}}
\newcommand{\g}{\gamma}
\newcommand{\Sgen}{S_\text{gen}}
\newcommand{\Sig}{\Sigma}

\newtheorem{remark}{Remark}

\makeatletter
\providecommand*{\cupd}{%
  \mathbin{%
    \mathpalette\@cupd{}%
  }%
}
\newcommand*{\@cupd}[2]{%
  \ooalign{%
    $\m@th#1\cup$\cr
    \hidewidth$\m@th#1\cdot$\hidewidth
  }%
}
\makeatother

\usepackage{hyperref}

\begin{document}
\preprint{MIT-CTP/6034}

\title{A Quantum Singularity Theorem for the Evaporating Black Hole}

\author{Netta Engelhardt}
\email{engeln@mit.edu}
\author{Ivri Nagar}%
 \email{ivri@mit.edu}

\affiliation{Center for Theoretical Physics -- a Leinweber Institute, Massachusetts Institute of Technology, \\Cambridge, MA 02139, USA}


\begin{abstract}
We prove a singularity theorem in semiclassical gravity without assuming global hyperbolicity or the null energy/curvature condition; the former is replaced by the weaker causality conditions of stable causality and past reflectivity, and the latter is replaced as is standard by the Generalized Second Law. This establishes in particular that the standard models of evaporating black holes are singular -- i.e. they are null geodesically incomplete.

\end{abstract}

\maketitle

\section{Introduction}\label{sec:intro}
The singularity theorems of general relativity are a cornerstone governing the behavior of gravity: generic initial data results in a breakdown of the classical theory, diagnosed by the incompleteness of null~\cite{Pen65, EicGal12, Sil12} or timelike~\cite{Haw66, HawPen70} geodesics. The structure of a typical ``black hole'' singularity theorem -- i.e. a theorem for future incompleteness in spatially open universes -- follows Penrose's original result; it involves three conditions: (1) a global causality condition ensuring general well-behavedness of the spacetime, (2) a curvature or energy condition, and (3) a condition on initial data (the existence of trapped surfaces)\footnote{Cosmological singularity theorems (e.g. for closed universes) have a different structure altogether, see~\cite{Min19b} for a review.}. The guarantee is then that some causal geodesics are both incomplete and inextendible, i.e. there are freely falling observers whose experience as described by the classical theory terminates at some finite proper time (or affine parameter). 

The breakdown of the classical theory is of significant interest in understanding its UV completion. This is particularly relevant in the context of black hole evaporation~\cite{Haw75}, where the predictions of the gravitational effective field theory deviate from the fundamental theory already at low curvatures~\cite{Pag93b}. While the presence or absence of curvature singularities in an evaporating black hole is inessential for the information puzzle, the end state of the evaporating black hole and the extent to which the radiation after evaporation encodes singularity physics is a promising avenue towards understanding the strong quantum gravity regime.

Despite the extensive literature on evaporating black holes, to the authors’ knowledge there is no result guaranteeing the existence of incomplete null or timelike geodesics in standard models of evaporating black holes. Such a theorem, if it were to follow the general structure above, would need: 
\begin{enumerate}
    \item A global causality condition weak enough to allow topology change of timeslices (i.e. a relaxation of the global hyperbolicity requirement of~\cite{Pen65});
    \item A matter or curvature condition that allows large violations of the null curvature condition (NCC), $R_{ab}k^{a}k^{b}\geq 0$ for all null vectors $k^{a}$, or equivalently under the semiclassical Einstein equation violations of the null energy condition (NEC), $\langle T_{ab}\rangle k^a k^b\geq 0$. These violations are essential for the Hawking evaporation process~\cite{DavFul76, For97};
    \item Initial data involving a notion of trapped surfaces which is sufficiently robust and stable to include large entropy gradients that compete with area gradients (weighted by $1/G$), without  a parametric separation between the two (e.g.~\cite{Pen19, AlmEng19}). 
\end{enumerate}
Significant progress on the front of (2) was made by Wall in~\cite{Wal10QST} by replacing trapped surfaces by ``quantum trapped surfaces’’ and the NCC/NEC by the Generalized Second Law (GSL)~\cite{Bek72, Bek74}. The latter is the statement that the generalized entropy $S_{\rm gen}(a)=\frac{{\rm Area}(\partial a)}{4G}+S_{\rm vN}(\rho_{a})$ increases monotonically in time when $a$ is the exterior of the event horizon on some timeslice. In general $S_{\rm gen}(a)$ is defined for any sufficiently well-behaved codimension-one achronal set $a$ (here $\rho_{a}$ is the state of quantum fields on $a$). There are also globally hyperbolic generalizations of Penrose's theorem using the averaged or smeared energy conditions or inequalities instead of the GSL~\cite{FewGal11, FewKon19, FreKon20, FreKon21,  FewKon22}. Wall’s result was further improved by Bousso in~\cite{Bou25}, which addressed point (3) (see also~\cite{BouSha22a, BouSha22b}). However, global hyperbolicity remains a \textit{sine qua non} for these results, rendering them inapplicable to the evaporating black hole that ends in a cloud of Hawking radiation~\cite{Kod79,His81}; see Fig.~(\ref{fig:EBHconf}). In fact, evaporating black holes may violate even weaker causality conditions such as causal continuity \cite{Les18}. 

\begin{figure}
    \includegraphics[width=0.6\columnwidth]{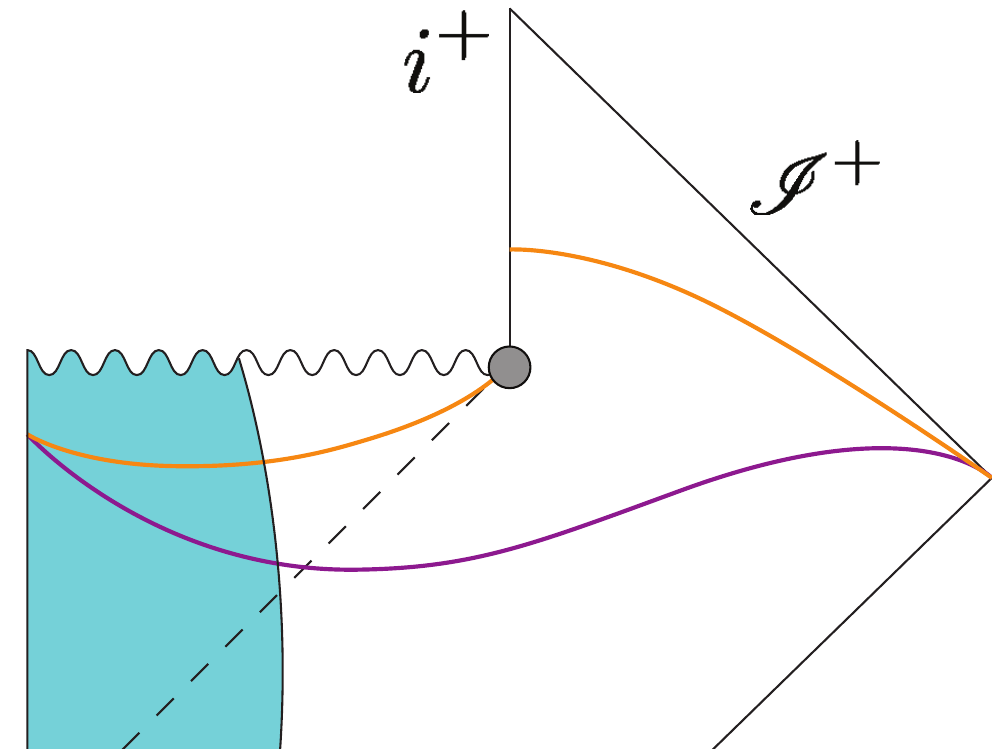}
    \caption{An evaporating black hole formed from collapse of (blue) matter, with an expected singularity (gray circle). Timeslices (purple and gold) change topology dynamically.}
    \label{fig:EBHconf}
\end{figure}

In parallel, some mathematical results have focused on relaxing causality conditions. In~\cite{Min19}, Minguzzi proved a singularity theorem assuming (1) past reflectivity and spatial openness (see below), (2) the NCC and (3) existence of a trapped surface. Minguzzi's conditions are entirely unrelated to global hyperbolicity; they are expected to hold for standard models of black hole evaporation~\cite{Min19b}.

A primary hurdle in translating Minguzzi's proof to the physical evaporating black hole is the former's heavy reliance on generator locality of the NCC to ensure pointwise focusing of null geodesics. In this paper we overcome this challenge: hybridizing the matter-appropriate approaches of Wall and Bousso with the causality-appropriate approach of Minguzzi, we prove a quantum singularity theorem for evaporating black holes. We assume: 
\begin{enumerate}
    \item Stable causality, past reflectivity, and spatial openness;
    \item The Generalized Second Law on future causally complete causal horizons\footnote{Since we are not assuming global hyperbolicity, we require the GSL specifically along generators that reach $\mathscr{I}^{+}$, see below.};
    \item The existence of a ``robustly quantum trapped surface'', closely related to the notion of trapped used in~\cite{Bou25}, with a compact inner region, defined more precisely below. 
\end{enumerate}
Under these assumptions, we prove that the spacetime is future null geodesically incomplete. This establishes that evaporating black holes typically studied in the literature indeed have singularities.

\section{Definitions and Assumptions}\label{sec:defs}

Throughout, let $(\cM,g)$ be a stably causal smooth spacetime with a future asymptotic boundary $\mathscr{I}^{+}$ or $i^{+}$. We list here the definitions and terminology we use, where these are new or differ from those in the literature (see e.g. \cite{Wald, HawEll}), or where a review or disambiguation is in order. 

\begin{itemize}
    
    \item  Our convention is that $x\in J^\pm(x)\backslash I^\pm(x)$. Note also that 
    $D(a)$ is open for any acausal hypersurface $a$ by~\cite{One83, Min12}. A \textbf{time function} is a $C^1$ function on $\cM$ that strictly increases along any future directed causal curve; a \textbf{timeslice} is a level set of such a function.

    \item A closed achronal set with empty edge is called a \textbf{slice}, and by \textbf{hypersurface} we shall mean a subset of a slice which is open in the slice topology (and in particular is achronal). $(\cM,g)$ is \textbf{past-reflecting} if $\forall p,q\in\cM \,, I^+(p)\sset I^+(q) \Rightarrow I^-(q)\sset I^-(p)$, and \textbf{spatially open} if it does not contain a compact spacelike hypersurface.

    \item A \textbf{future directed outward null deformation} of a hypersurface $a$ within an open set $O$, is a hypersurface $b$ obeying $D(a)\sset D(b)$, $b\sset \cl(I^+(a))$, $a\Delta b\sset O$, and $\edge(b)\sset\edge(a)\cup[(\partial I^+(a)\backslash a)\cap O]$. Here and throughout, $\Delta$ denotes the symmetric difference, $\partial(\cdot)$ the topological boundary, and $\cl(\cdot)$ represents the topological closure. The definition with ``future'' replaced by ``past'' is analogous.

    \item A hypersurface $a$ is \textbf{entropy normal} at $p\in \edge(a)$ if there exists a neighborhood (for us, a neighborhood is always open) $O$ of $p$ such that any past directed outward null deformation $b$ of $a$ within $O$ has $\Sgen(b)-\Sgen(a)\leq 0$. $a$ is \textbf{stably entropy-anomalous} if there exists a neighborhood $O$ of $\edge(a)$ such that for any proper future directed outward null deformation $b$ of $a$ in $O$, there is a new edge point $p\in \edge(b)\backslash\edge(a)$ at which $b$ is not entropy normal\footnote{
    These definitions are similar to those of \cite{Bou25}, where entropy normal was called ``past-nonexpanding'' and stably entropy-anomalous ``past-noncontracting''. The deformations allowed by the definitions here are more restricted than those of \cite{Bou25}, but this restriction does not hinder their ability to capture all the quasi-local behavior. Note also that~\cite{Bou25} used one-shot entropies; we expect that our results can be generalized to this case.
    }. (Here, ``proper'' means $\Int(D(b))\neq\Int(D(a))$.) 

    \item A hypersurface is \textbf{robustly quantum trapped} if it is stably entropy-anomalous and has compact edge.
\end{itemize}

Intuitively, entropy normality is the robust quantum analogue of outwards future expansion, and stable entropy-anomalousness is a quantum version of stably trapped: $\edge(a)$ should be bounded away from being a normal surface.

\begin{itemize}
    \item A \textbf{causal horizon} is a set of the form $\partial I^-(W)$ for some causal worldline $W$ that reaches the future asymptotic boundary of spacetime. 

    Note that by \cite{HawEll}, since the spacetime is stably causal, a causal curve cannot be imprisoned: it reaches the future asymptotic boundary if and only if it is future-infinite in affine parameter. 
    
    \item An \textbf{exterior} of the horizon $\partial I^-(W)$ is a hypersurface $a\sset \cl(I^-(W))$ with $\edge(a)\sset \partial I^-(W)$.
\end{itemize}

In addition, we rely on two conjectures concerning the behavior of $\Sgen$. The first is the GSL, which in a globally hyperbolic spacetime states that $\Sgen$ of the horizon exterior is nondecreasing with time. When the spacetime is stably causal but not globally hyperbolic, there are two potential issues: first, timeslices may dynamically change topology; the states of quantum fields on such slices will generally not be unitarily related in effective field theory. In this case the GSL should not facilitate a comparison between them. We say that two sets $A,B\sset\cM$ are ``Cauchy comparable'' if they are contained in some mutual globally hyperbolic sub-spacetime. Second, a null geodesic on the horizon may fail to reach future null infinity. In this case, a future deformation of the exterior along this geodesic may not increase $\Sgen$. With these two points in mind, we state the relevant version of the GSL:
\begin{conjecture}[GSL]\label{GSLgen}
    Let $W$ be a future-infinite causal worldline, and let $a, b$ be two Cauchy comparable exteriors of the horizon $\partial I^-(W)$ such that $D(a)\sset D(b)$. Suppose in addition that through every point in $\edge(b)\backslash\edge(a)$ there passes a null geodesic $\g\sset \partial I^-(W)$ which is future-infinite. Then $\Sgen(a)-\Sgen(b)\geq 0$.
\end{conjecture}

In fact, in our stably causal spacetime, the second condition will only fail due to an incomplete null geodesic. Since our aim is precisely to prove the existence of such a singularity, this condition is unnecessary for our purposes, as formalized in Cor. (\ref{cor:GSLnonsing}). 

The second conjecture is strong subadditivity, which holds at first subleading order written out explicitly above, but is yet unproven for the full quantity.

\begin{conjecture}[Strong subadditivity]
    For achronal codimension-one $a\sset b$ and $c$, with $b,c$ spacelike to each other, $\Sgen(b\cup c)-\Sgen(a\cup c)\leq \Sgen(b)-\Sgen(a)$.
\end{conjecture}

In the remainder of this work, we assume both these conjectures hold without further comment.

\section{Singularity Theorem}\label{sec:thm}

We now state our main result:

\begin{theorem}[Causally Robust Quantum Singularity Theorem]\label{thm:main}
    Let $(\cM,g)$ be a stably causal, past-reflecting, and spatially open spacetime. Let $\Sig_a$ be a  timeslice, and let $a\sset\Sig_a$ be an open subset thereof, which is robustly quantum trapped. Furthermore, suppose that $\edge(a)$ is nonempty, $C^1$, and divides $\Sig_a$ into an exterior $a$ and an interior $a'=\Sig_a\backslash (a\cup \edge(a))\neq \varnothing$, such that $a'\cup\edge(a)$ is compact. Assume the GSL and strong subadditivity of $\Sgen$. Then $(\cM, g)$ is future null geodesically incomplete.
\end{theorem}

The proof appears at the end of the section after the requisite intermediate results, whose proofs are relegated to the Appendix. We expect that the compactness requirement of $a'\cup \edge(a)$ can be relaxed. 

\begin{restatable}{prop}{edgeprop}\label{prop:edge}
    For achronal $\Sig$ and $a\sset\Sig$, it holds that $\partial_\Sig a\sset \edge(a)\sset\partial_\Sig a\cup\edge(\Sig)$.
\end{restatable}

\begin{remark}\label{rem:edge}
     The edge is thus the boundary in the hypersurface topology: all continuations of $a$ to a slice $\Sig$ yield the same $\edge(a)=\partial_\Sig a$.
\end{remark}

The subscript $\Sig$ means that we take the boundary or closure within $\Sig$ as a topological subspace. 

The central result we use to relax global hyperbolicity is \cite{Min19}'s Theorem (2.10). We quote the theorem here as it applies to our case.

\begin{prop}[Corollary of Theorem \text{(}2.10\text{)} of \cite{Min19}]\label{prop:Ming}
    Suppose $(\cM,g)$ is past-reflecting and spatially open. Let $ S\sset \cM$ be compact and nonempty. Then there exists a future inextendible null geodesic starting from $S$ that never enters $I^+(S)$.
\end{prop}

The GSL compares timeslices that may differ exclusively in the neighborhood of some point. Intuitively, given any timeslice, we may locally deform it by ``pushing it forward in time'' in a small region so that it stays a timeslice. The following proposition formalizes this intuition.

\begin{restatable}[Timeslices can be deformed]{prop}{timedeformprop}\label{prop:timedeform}
    Let $\Sig$ be a timeslice, $X\sset \Sig$ be compact and $U\supseteq  X$ be open. Then there exists a Cauchy evolution of $\Sig$ forward in time within $U$ to a new timeslice (with an equally smooth corresponding time function) $\Sig'$ such that $\Sig'\cap\Sig\cap X=\varnothing$.
\end{restatable} 

In the introduction, we alluded to a simpler operational statement of the GSL, which requires the result formalized below.

\begin{prop}
    Suppose that the spacetime $(\cM,g)$ is future null geodesically complete. Let $W\sset \cM$ be a future-infinite causal worldline, and $p\in\partial I^-(W)$. Then there exists a future-infinite null geodesic $\g\sset \partial I^-(W)$ passing through $p$.
\end{prop}
\begin{proof}
    Observe that $W$ can only fail to be closed by not containing its past endpoint, in which case we can add it to $W$ without changing $\partial I^-(W)$, so that without loss of generality we take $W$ to be closed. Then, by Thm. (8.1.6) of \cite{Wald}, either (i) $p\in W$, or there exists a null geodesic $p\in \g \sset \partial I^-(W)$, future directed and starting at $p$, which is either (ii) future inextendible or (iii) has a future endpoint $p'\in W$. We consider each case in turn.
\renewcommand{\labelenumi}{\roman{enumi}}
\begin{enumerate}
    \item This implies that, after $p$, $W$ is a null geodesic (otherwise $W$ would enter its chronological future, and then $p\in I^-(W)$); note $W$ is future-infinite.
    \item By the assumption of future null geodesic completeness, $\g$ is future infinite.
    \item The concatenation of $\g$ with $W$ at $p'$, as in (i), is a null geodesic, and future-infinite.
\end{enumerate}

\end{proof}
\begin{corollary}[GSL assuming future null geodesic completeness]\label{cor:GSLnonsing} 
    If the spacetime is future null geodesically complete, then the final condition in the GSL (Conj. (\ref{GSLgen})) becomes redundant, and the GSL may be equivalently restated without it: Let $W$ be a future-infinite causal worldline, and let $a, b$ be two Cauchy comparable exteriors of the horizon $\partial I^-(W)$ such that $D(a)\sset D(b)$. Then $\Sgen(a)-\Sgen(b)\geq 0$.
\end{corollary}

Now, the GSL guarantees a monotonic increase of $\Sgen$ to the future along a future causal horizon.  It follows almost by definition that horizon exteriors must be entropy normal:

\begin{restatable}[GSL implies entropy normal exteriors]{lem}{GSLPNElem}\label{lem:GSLPNE}
    Suppose that the spacetime $(\cM,g)$ is future null geodesically complete. Let $W$ be a future-infinite causal worldline and $\Sig$ be a slice with $\Sig\sset\Int(D(\Sig))$. Then $a=I^-(W)\cap \Sig$ is an exterior of the horizon $\partial I^-(W)$, and entropy normal at all its edge points.    
\end{restatable}

Finally, the following lemma captures the intuition that a hypersurface being entropy normal at a point is a local property -- if a hypersurface is entropy normal at a point, then enlarging it elsewhere shouldn't change this, as long as both the hypersurface itself and its lightcone near the point in question are unaffected.
\begin{restatable}[Locality of entropy normality]{lem}{PNELoc}\label{lem:PNELoc}
    Suppose that the spacetime $(\cM,g)$ is future null geodesically complete. Let $\Sig$ be a slice, and $a_1,a_2\sset \Sig$ be hypersurfaces. Denote the union hypersurface $a=a_1\cup a_2$, and let $p\in \edge(a_1)\backslash \cl(I(a_2\backslash a_1))$. Then $p\in\edge(a)$, and if $a_1$ is entropy normal at $p$, so is $a$.
\end{restatable}

The proof of Thm. (\ref{thm:main}), which is a hybrid of the proofs of~\cite{Bou25, Min19}, is somewhat obscured by technical issues, so we give an outline of it before jumping in, see Fig (\ref{fig:proofsketch}). 

\begin{figure}
    
    \includegraphics[width=0.9\linewidth]{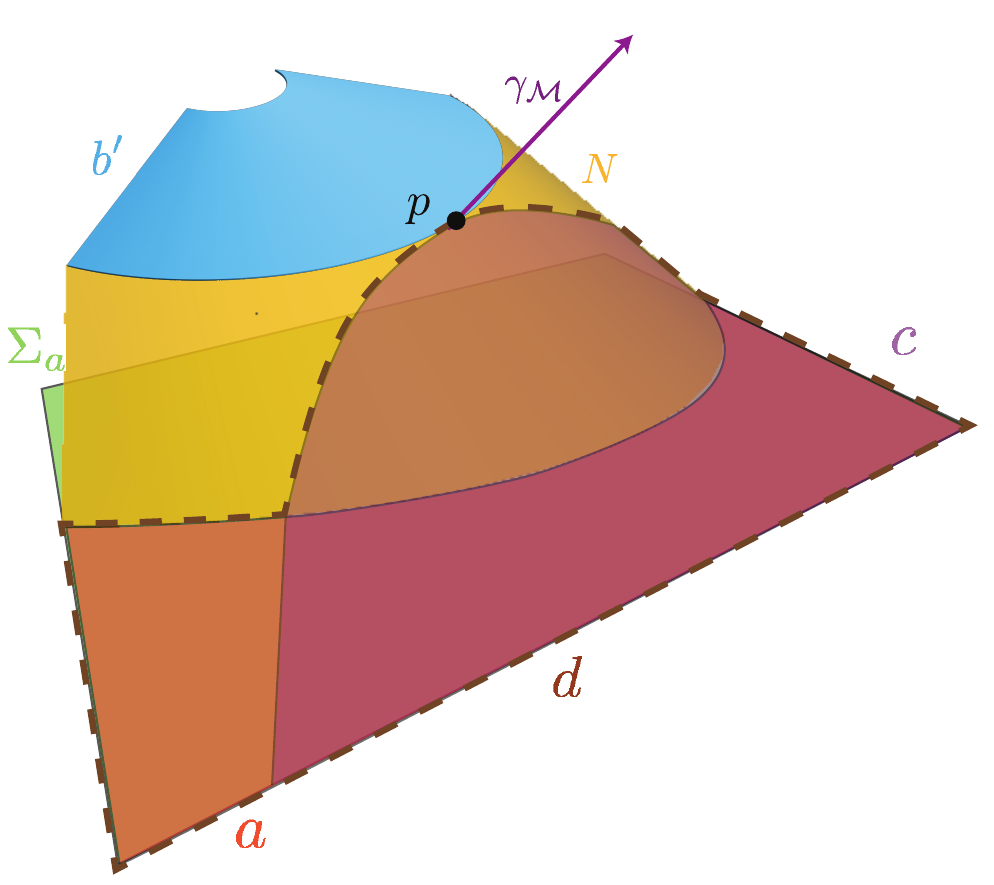}
    \caption{A diagram of the construction in the proof. $b'$ (blue) is on a timeslice to the future of $\Sig_a$ (green). A null congruence $N$ (gold) is fired outwards from $a$ (red). There is a future-infinite null ray $\g_\cM$ (purple) emanating from some point $p$ (black) on the edge of $b'$. $c$ (purple shade) is the intersection of the past of this ray with the surface, and $d=a\cup c$ (inside dashed brown line) is then both a small deformation of $a$ and a horizon exterior.}
    \label{fig:proofsketch}
\end{figure}

The idea is to construct a horizon, for which exteriors are entropy normal, near $\edge(a)$. To do so, we first push $\Sig_a$ slightly forward in time in the vicinity of $\edge(a)$ to a new timeslice $\Sig_b$. Since the portion of $\Sig$ which is not in $a$ is compact, the portion of the new timeslice which is not in $I(a)$ is also. Viewing $a$ as the exterior part of $\Sig_a$, this construction amounts to firing a null congruence from $\edge(a)$ into the interior, tracking it for some short time, and focusing on the new interior $b'$. Next, we use the compactness of $b'$ to conclude from Prop. (\ref{prop:Ming}) that there is a future inextendible null geodesic $\g_\cM$ that doesn't enter $I^+(b')$. We then proceed by contradiction and assume future null geodesic completeness, so that this geodesic generates a horizon $\partial I^-(\g_\cM)$. In the third section we focus one of this horizon's exteriors $c$ near $a$, and consider $d=a\cup c$. This leads to the following tension. Since $d$ is a small deformation of $a$, the stably entropy-anomalous quality of $a$ implies that $d$ fails to be entropy normal at some point on the edge of $c$, that is, somewhere along the horizon. However, the GSL implies that $d$ must be entropy normal everywhere along the horizon. This yields our ultimate contradiction.

\begin{proof}[Proof of Thm. (\ref{thm:main})]
\phantomsection\label{proof:main}
The proof is partitioned into three sections as outlined above.

\vspace{0.1cm}
\noindent\textit{Constructing $b'$:}

Denote the globally hyperbolic subspacetime $\cM'=\Int(D(\Sig_a))=D(\Sig_a)$. Let $O_*$ be the neighborhood of $\edge(a)$ guaranteed by $a$ being stably entropy-anomalous. There exists a small neighborhood $U_0$ of $\edge(a)$ obeying that (1) $\cl(U_0)\sset O_*\cap\cM'$; and (2) every null geodesic fired to the future and outwards from $a$ starting on $\edge(a)$, if extended maximally, intersects $I^+(a'\backslash U_0)$, whose construction is detailed in the Appendix. Note in particular $a'\backslash U_0\neq \varnothing$.

We now apply Prop. (\ref{prop:timedeform}) with $\edge(a)\sset \Sig_a$ and $U_0$ to obtain a new timeslice $\Sig_b$. Denote the region between $\Sig_a$ and $\Sig_b$ (inclusive) by $R=J^+(\Sig_a)\cap J^-(\Sig_b)$. Henceforth in this proof, it is convenient to work only within $\cM'$, and we do so unless stated otherwise. Now, $D^+(a'\cup\edge(a))$ is closed, disjoint from $I^+(a)$, and $I^+(\Sig_a)\sset I^+(a)\cup D^+(a'\cup\edge(a))$ (see e.g.~\cite{BouEng15b}). Set $b'=\Sig_b\cap D^+(a'\cup\edge(a))$. 

\vspace{0.1cm}
\noindent\textit{Constructing $c$:}

First, note that $a'\backslash U_0\sset b'$, and in particular $b'$ is nonempty. Moreover, $b'$ is compact because $\Sig_b$ and $\Sig_a$ are homeomorphic as Cauchy slices of $\cM'$. The image of $b'$ under the relevant time flow is contained in  $a'\cup\edge(a)$. Since $D^+(a'\cup\edge(a))$ is closed, $b'$ is closed, so its image is closed, and therefore its image is compact, so  $b'$ itself is compact.

Set $N=\partial I^+(a)\cap R$. Note $a\sset N$. Also, note that by the construction of Prop. (\ref{prop:timedeform}), $R\backslash\Sig_a\sset U_0$, so in particular $N\backslash a\sset O_*$. Set $\Sig=b'\cup N$, which we now prove is a Cauchy slice for $\cM'$. We begin by showing achronality. $\Sig_b$ is achronal as a timeslice, and $\partial I^+(a)$ is also achronal, so $b'$ and $N$ are both achronal as subsets thereof. Assume for contradiction that there is a timelike curve from $q\in b'$ to $q'\in N$. $q'\in R\sset J^-(\Sig_b)$ and $\Sig_b$ is achronal, so $q'\notin I^+(q)$. Therefore $q'\in I^-(q)$, and so there is a neighborhood $O\sset I^-(q)$ of $q'$. Since $q'\in \partial I^+(a)$ there is a point $q''\in O\cap I^+(a)$, whence $q\in I^+(a)$, contradicting $q\in b'\sset D^+(a'\cup\edge(a))$.

To now show that $\Sig$ is a Cauchy slice, let $q\in D(\Sig_a)=D(\Sig_b)$, and consider an inextendible causal curve $\g$ through it. This curve must pass through some points $q_a\in\Sig_a, q_b\in\Sig_b$, and possibly one of these is in $\Sig$; suppose $q_a, q_b\notin \Sig$. Then $q_b\notin\Sig_a$ (as $\Sig_b\cap\Sig_a\sset \Sig$), so $q_b\in I^+(\Sig_a)$ (as $\Sig_b\sset\Sig_a\cup I^+(\Sig_a)$), and so $q_b\in I^+(a)$ (as $\Sig_b\cap D^+(a'\cup\edge(a))\sset \Sig$). Since $\g$ goes from $q_a\notin I^+(a)$ to $q_b\in I^+(a)$, it must cross the boundary at some point $q'\in \partial I^+(a)$ between them. $q'\in R$ so $q'\in N\sset \Sig$. Thus, in any case $\g$ intersects $\Sig$. This demonstrates $q\in D(\Sig)$, meaning $D(\Sig_a)\sset D(\Sig)$; moreover, $\Sig\sset D(\Sig_a)$ so $D(\Sig)\sset D(\Sig_a)$. In total, $D(\Sig)=D(\Sig_a)=\cM'$, so $\Sig$ is a Cauchy slice of $\cM'$.

By Prop. (\ref{prop:Ming}), there exists a future inextendible null geodesic $\g_\cM$ in $\cM$, emanating from some $p\in b'$, which never enters $I^+(b')$. Now assume for contradiction that $\cM$ is future null geodesically complete. Then $\g_\cM$ is extended with infinite future affine parameter and defines a horizon $\partial I^-_\cM(\g_\cM)$ -- note that the subscript indicates that this set is determined by the chronological relations in the full spacetime $\cM$, rather than $\cM'$. Define the hypersurface $c=I^-_\cM(\g_\cM)\cap \Sig =I^-_\cM(\g_\cM)\cap N$ within $\Sig$. 
Note that $a$ is also a hypersurface on $\Sig$, since $\partial_\Sig a=\edge(a)=\partial_{\Sig_a} a$ is disjoint from $a$, allowing us to define the hypersurface $d=c\cup a$ within $\Sig$.

\vspace{0.1cm}
\noindent\textit{Analyzing $d$:}

Basic topology on $\Sig$ implies $\edge(d)\backslash\edge(a)\sset\edge(c)$, and we further claim that $p\in \edge(d)\backslash\edge(a)$. To demonstrate this, first note that $p\notin\edge(a)$ by the construction of $\Sig_b$. $b'$ is disjoint from both $a$ and $I^-_\cM(\g_\cM)$, so $p\in b'$ implies $p\notin d$. Now, let $\Phi$ be the mapping within $\cM'$ onto $\Sig$ according to the integral flow of some time function. Observe a sequence of points $\g_\cM\ni p_n\to p$. From some point onwards $p_n\in\cM'$, and then $\Phi(p_n)$ is defined and we have $c\ni\Phi(p_n)\to\Phi(p)=p$. Thus $p\in \cl(c)\sset\cl(d)$, so that (following Remark (\ref{rem:edge})) in total $p\in\edge(d)\backslash\edge(a)$. 

We now establish that $d$ is a proper future directed outward null deformation of $a$ within $O_*$. It is immediate that $d\sset \cl(I^+(a))$, $d \Delta a\sset O_*$, as well as $a\sset  d$ and so $D(a)\sset D(d)$. To establish the edge criterion, note $\edge(d)\sset N$ since $d\sset N$ and $N$ is closed. Furthermore, $\edge(d)\cap a=\varnothing$ since $a$ is open in $\Sig$, so $\edge(d)\sset N\backslash a$. $\cl(d)\sset \cl(U_0)\sset O_*$ by the construction of $U_0$, so that together $\edge(d)\sset( N\backslash a)\cap  O_*$. 
To show the deformation is proper, we again use the mapping $\Phi$ to $\Sig$. By continuity, since $\Phi(p)\notin \cl(a)$ (as $p\in b'$) there is a point $q\in \g_\cM\backslash{p}$ with $\Phi(q)\notin \cl(a)$. Consider the open set $O=I^-(q)\cap I^+(\Sig_b)\cap \Phi^{-1}(\Sig\backslash \cl(a))\neq \varnothing$. Since $q\in I^+(\Sig_b)$, any point (sufficiently close to $q$) along the integral flow line of $q$ backwards in time , will be in $O$. Let $q'\in O$. Then the inextendible integral flow line of $q'$ intersects $\Sig\backslash\cl(a)$, and so by achronality of $\Sig$ doesn't intersect $a$, showing $q'\notin D(a)$. However, any inextendible causal curve $\g$ through $q'$ passes through $\Sig$ at some point $q''$ to the past of $q'$. 

Since $q'\in I^-(q)$ and $q\in\g_\cM$, the point $q''\in\Sig$
where any inextendible causal curve through $q'$ meets $\Sig$
lies in $I^-_\cM(\g_\cM)\cap\Sig=c\sset d$, so $q'\in D(d)$. In total, since $O$ is open, this shows $ O\sset \Int(D(d))\backslash\Int(D(a))$, concluding the verification.

Lastly, we claim that $(\edge(d)\backslash\edge(a))\cap\cl(I(a\backslash c))=\varnothing$. Assume for contradiction that $p'\in(\edge(d)\backslash\edge(a))\cap\cl(I(a\backslash c))$. Since $p'\in N\backslash\Sig_a$, in particular $p'\in \cl(I^+(a\backslash c))$. Thus there exist future directed timelike curves $\g_n$ from $q_n\in a\backslash c$ to $p_n$, with $p_n\to p'$. By the limit curve theorem, there exists a past inextendible causal curve through $p'$ which is a limit curve thereof. This curve must intersect $\Sig_a$ at some point $p''$; we label by $\g$ the portion of the curve between $p'$ and $p''$. Since $\g_n$ intersect $a\backslash c\sset \Sig_a$, they are disjoint from $D(\Sig_a\backslash(a\backslash c))$, and therefore the limit point obeys $p''\notin \Int(D(\Sig_a\backslash(a\backslash c)))$. 

The set $\Sig_a\backslash\cl(a\backslash c)=\Int_{\Sig_a}(\Sig_a\backslash(a\backslash c))$ is open in $\Sig_a$, so $D(\Sig_a\backslash\cl(a\backslash c))$ is open. Since a set is contained in its domain of dependence, this implies $\Sig_a\backslash\cl(a\backslash c)\sset \Int(D(\Sig_a\backslash(a\backslash c)))$, whence $p''\in \cl(a\backslash c)$. $\g\sset \cl(I^+(a))$ and $p'\notin I^+(a)$, which implies $\g\sset N\backslash a$ and in particular $p''\in\edge(a)$. 

Let $\gamma$ be a null generator of $\partial I^-_\cM(\g_\cM)$ starting at $p'$. It is either future inextendible or ends on $\g_\cM$. In the latter case, we concatenate it with $\g_\cM$ at its future endpoint, rendering it a future inextendible causal curve. We further concatenate with $\g$ at $p'$, and label the entire curve $\g'$. Since $p''\notin I^-_\cM(\g_\cM)$, $\g'$ must be achronal, and therefore it is a null geodesic with $\g'\sset\partial I^-_\cM(\g_\cM)$. That is, $\g'$ is a future inextendible, future directed null geodesic emanating from $p''\in\edge(a)$. It initially (i.e. near $p''$) resides on $N\backslash a$, meaning it was fired outwards from $a$ starting at $\edge(a)$. Therefore, by construction of $U_0$, it must intersect $I^+(a'\backslash U_0)\sset I^+(b')$. But then $\g'\sset\partial I^-_\cM(\g_\cM)$ would imply that in fact $\g_\cM$ also intersects $I^+(b')$, which is a contradiction.

\vspace{0.1cm}
\noindent\textit{Ultimate contradiction:}

Since $a$ is stably entropy-anomalous, and $d$ is a proper future directed outward null deformation of $a$ within $O_*$, there exists a new edge point $p'\in \edge(d)\backslash\edge(a)$ at which $d$ is not entropy normal. $p'\in \edge(c)$, so by Lem. (\ref{lem:GSLPNE}) on exteriors being entropy normal, $c$ is entropy normal at $p'$. Additionally, $p'\notin \cl(I(a\backslash c))$, so by Lem. (\ref{lem:PNELoc}) on the locality of entropy normality with $a_1=c$, $a_2=a$, we conclude that $d$ is entropy normal at $p'$. This is a contradiction.

\end{proof}

\section*{Acknowledgments} The work of NE is supported in part by the Department of Energy under Early Career Award DE-SC0021886 and the HEP-QIS program under grant number DE-SC0025937, by the Heising-Simons Foundation under grant no. 2023-4430, by the Moore Foundation via the Black Hole Initiative, and by the MIT Physics Department. The work of IN is supported by the MIT Physics Department. 

\appendix
\section{Appendix: Assorted Proofs}\label{app}

\edgeprop*
\begin{proof}\phantomsection\label{proof:edge}
    For the first inclusion, let $p\in \partial_\Sig a$. $p\in \cl_\Sig(a)$ so since $a\subset \Sigma$, $p\in\cl(a)$. Proceed by contradiction: assume that $p\notin \edge(a)$, and let $ O$ be some neighborhood of $p$ on which the edgeness condition is violated. Let $q\in I^+_O(p)$ and $r\in I^-_O(p)$. Then $p\in I^-_O(q)$ and so there is a neighborhood $O_q\sset I^-(q)$ of $p$, and similarly a neighborhood $O_r\sset I^+_O(r)$ of $p$. Consider the neighborhood $O'=O_r\cap O_q\cap\Sig$ of $p$ in $\Sig$. Let $p'\in O'$, then there is a timelike curve within $O$ from $r$ to $q$ passing through $p'$. Since $\Sig$ is achronal, this curve cannot intersect $a$ at any other point, and so $p'\in a$.  This shows $p\in\Int_\Sig (a)$, a contradiction, and so proves the first inclusion. 

For the second, let $p\in\edge(a)\backslash\edge(\Sig)$. This means $p\in\Sig$, so $p\in\cl_\Sig(a)$. Assume for contradiction that $p\in\Int_\Sig (a)$, so there is a neighborhood $O''$ of $p$ with $O''\cap\Sig\sset a$. Let $O'$ be the neighborhood of $p$ guaranteed by $p\notin\edge(\Sig)$, and denote $O=O''\cap O'$. Let $q\in I^+_O(p)$ and $r\in I^-_O(p)$, then any timelike curve between them within $O$ intersects $\Sig$ at some $p'$. Then $p'\in a$, so $O$ demonstrates $p\notin\edge(a)$, yielding the desired  contradiction. 
\end{proof}

Before proving Prop. (\ref{prop:timedeform}), we start with a localized version of it, for deformations away from some point. Iterating this, we can perturb general compact sets.

\begin{lem}
    Let $\Sig$ be a timeslice. Let $p\in \Int(D(\Sig))=D(\Sig)$, and choose a local coordinate neighborhood $U\supseteq p$ with coordinates $\phi:U\to \mathbb{R}^d$, where $d$ is the spacetime dimension. Let $\ell>0$ be sufficiently small that $B_{p,\ell}:=\phi^{-1}(B(\phi(p),\ell))\sset  D(\Sig)$, where $B(x,r)\sset\mathbb{R}^d$ is the open Euclidean ball of radius $r\geq 0$ centered on $x\in\mathbb{R}^d$, and let $0<\ell'<\ell$. Then there exists a Cauchy evolution of $\Sig$ forward in time within $B_{p,\ell}$ to a new timeslice $\Sig'$, such that $\Sig'\cap\Sig\cap B_{p,\ell'}=\varnothing$.
\end{lem}

\begin{proof}
Let $T:\cM\to\mathbb{R}$ be a time function with $\Sig=T^{-1}(0)$ and timelike $\nabla^{a}T$ by~\cite{BerSan04}. Shrinking $\ell$ if necessary (but still with $\ell>\ell'$), assume $\cl(B_{p,\ell})\sset U\cap D(\Sig)$. Choose $\ell''=(\ell+\ell')/2$, and let $h$ be a smooth positive bump function supported in $B_{p,\ell''}$, $h=1$ on $B_{p,\ell'}$.

Since the timelike cone is open~\cite{One83} and $\operatorname{supp}h$ is compact, for sufficiently small $\epsilon>0$, $\widetilde T=T-\epsilon h$ is again a time function (as can be shown by bounding the gradients of $T$ and $h$).
Set $\Sig'=\widetilde T^{-1}(0)$. Then $\Sig\Delta\Sig'\sset B_{p,\ell}$, $\Sig'\cap\Sig\cap B_{p,\ell'}=\varnothing$, and $\Sig'$ lies nowhere to the past of $\Sig$.

It remains only to note that $D(\Sig')=D(\Sig)$. Any inextendible causal curve meeting one of $\Sig,\Sig'$ outside $B_{p,\ell''}$ meets the other there. If the intersection lies inside $B_{p,\ell''}\sset D(\Sig)$, then the curve meets $\Sig$ by definition; tracking the curve to the future until it exits $\cl(B_{p,\ell''})$, by continuity of $\widetilde T$ it must intersect $\Sig'$. Thus every inextendible causal curve intersects $\Sig$ iff it intersects $\Sig'$.
\end{proof}

\timedeformprop*
\begin{proof}\phantomsection\label{proof:timedeform}
    Since $D(\Sig)$ is open and contains $\Sig$, for each $p\in X$ there exists a local coordinate neighborhood $U_p\supseteq p$ with coordinates $\phi_p:U_p\to \mathbb{R}^d$ and an $\ell_p>0$ such that $B_{p,\ell_p}:=\phi_p^{-1}(B(\phi_p(p),\ell_p))\sset D(\Sig)\cap U$. Then $\{B_{p,\ell_p/2}:p\in X\}$ is an open cover of $X$, so by compactness there is a finite subcover $B_{p_i,\ell_{p_i}/2}, i=1,...,n$. By applying the above lemma on each $p_i$ sequentially, with $\ell=\ell_{p_i}$ and $\ell'=\ell_{p_i}/2$, we obtain the desired result. Indeed, composing Cauchy evolutions forward in time yields another forwards Cauchy evolution; via the above construction, if $q\in\Sig$ is not in one of the intermediate timeslices then it won't be in the final one.
\end{proof}

\GSLPNElem*

\begin{proof}\phantomsection\label{proof:GSLPNE}
    By Prop. (\ref{prop:edge}) and Remark~\ref{rem:edge}, $\edge(a)=\partial_\Sig a$; it follows that $a$ is an exterior of the horizon $\partial I^-(W)$ immediately by definition. Set $O=\Int(D(\Sig))$; we proceed by showing that $(\partial I^-(a)\backslash a)\cap O\sset\partial I^-(W)$. 
    
    Let $q\in (\partial I^-(a)\backslash a)\cap O$. Since $q\in \cl(I^-(a))\sset \cl(I^-(W))$, it just remains to show $q\notin I^-(W)$. By contradiction: assume $q\in I^-(W)$.  Parametrize $W$ as $W(\lambda)$ for $\lambda\in\mathbb{R}$ sufficiently large. Then $q\in I^-(W(\lambda_0))$ for some $\lambda_0$. Since $W$ is causal, we have $q\in I^-(W(\lambda))$  for all $\lambda\geq\lambda_0$. There exist future directed timelike curves $\g_\lambda$ from $q$ to $W(\lambda)$, which we maximally extend. 
    $q\in O$ so each $\g_\lambda$ must intersect $\Sig$ at some point $q_\lambda$.
    
    If $q_\lambda$ is to the past of $W(\lambda)$ for some $\lambda\geq \lambda_0$, then $q_\lambda\in \Sig\cap I^-(W)=a$ and so $q\in I^-(a)\cup a\cup I^+(a)$ (corresponding to $q_\lambda$ to the future, at and to the past of $q$, resp.). This contradicts $q\in(\partial I^-(a)\backslash a)$ by stable causality. Otherwise, $q_\lambda$ is at or to the future of $W(\lambda)$ for all $\lambda\geq\lambda_0$, whence $W\sset J^-(\Sig)$ and so $I^-(W)\sset I^-(\Sig)$. Since $\Sig$ is achronal, $a=I^-(W)\cap \Sig=\varnothing$ which again contradicts $q\in(\partial I^-(a)\backslash a)$. We conclude $q\notin I^-(W)$, and thus $(\partial I^-(a)\backslash a)\cap O\sset\partial I^-(W)$.  

    Now, let $p\in\edge(a)$. $p\in \Sig$ so $O$ is a neighborhood of $p$. Let $b$ be past directed outward null deformation of $a$ within $O$. By the above we conclude that $\edge(b)\sset \partial I^-(W)$. Moreover $b\sset \cl(I^-(a))\sset \cl(I^-(W))$, meaning that $b$ is an exterior of $\partial I^-(W)$. The GSL yields $\Sgen(b)-\Sgen(a)\leq 0$, proving $a$ is entropy normal at $p$.
\end{proof}

\PNELoc*
\begin{proof}\phantomsection\label{proof:PNELoc}
    That $p\in\edge(a)$ follows immediately from $p\in \edge(a_1)\backslash \cl(I(a_2\backslash a_1))$, since $a_1,a_2$ are open in $\Sig$ and $\edge(\cdot)=\partial_\Sig(\cdot)$. Let $O'$ be the neighborhood of $p$ guaranteed by $a_1$'s entropy normality there. Denote $O''=O'\backslash\cl(I(a_2\backslash a_1))$, which is also a neighborhood of $p$. Let $O$ be a neighborhood of $p$ with $\cl(O)\sset O''$ (one such always exists: choose coordinates around $p$, find an open ball contained in $O''$, and take $O$ to be a concentric ball of half the radius). Let $b$ be a past directed outward null deformation of $a$ within $O$.  We now show that $b'=b\backslash(\cl(a_2\backslash a_1))$ is a past directed outward null deformation of $a_1$ within $O$, verifying this trait by trait. 

    First, $b'$ is a hypersurface. To prove $D(a_1)\sset D(b')$, let $q\in D(a_1)$ and $\g \ni q$ be an inextendible causal curve intersecting $a_1$ at some point $q'$. Possibly $q'\in b'$; suppose not. Then $q'\notin b$, since $a_1$ being open in $\Sig$ implies $b\backslash b'\sset b\backslash a_1$, so $q'\in a\backslash b\sset O$. Thus $q'\notin\cl(I(a_2\backslash a_1))=\cl(I(\cl(a_2\backslash a_1)))$ and $\g$ cannot intersect $\cl(a_2\backslash a_1)$. $q\in D(a_1)\sset D(a)\sset D(b)$, so $\g$ intersects $b$ at some point $q''\in b\backslash\cl(a_2\backslash a_1)=b'$. In either case $\g$ intersects $b'$, so $D(a_1)\sset D(b')$. To prove $b'\sset \cl(I^-(a_1))$, let $q\in b'\backslash a_1\sset b\backslash a$. Then $q\in \cl(I^-(a))$ because $q\in b$. $q\in O$ because $q\in a\Delta b$, so $q\notin \cl(I(a_2\backslash a_1))$; together with $q\in \cl(I^-(a))$ this implies $q\in \cl(I^-(a_1))$. Since $a_1\sset \cl(I^-(a_1))$, $b'\sset \cl(I^-(a_1))$. The above showed $b'\backslash a_1\sset O$. Moreover $a_1\backslash b'\sset a\backslash b\sset O$, so $a_1\Delta b'\sset O$. 
    
    Finally, to prove the edge requirement, we use the open regions $O''$ and $\cM\backslash \cl(O)$, which together cover the entire spacetime, $O''\cup(\cM\backslash \cl(O)) =\cM$. Within $O''$,  this region is disjoint from $\cl(I(a_2\backslash a_1))$ and in particular from $\cl(a_2\backslash a_1)$. We have $b\cap O''=b'\cap O''$, $a\cap O''=a_1\cap O''$ and $I^-(a)\cap O''=I^-(a_1)\cap O''$; since edges and boundaries are local constructs, $\edge(b')\cap O''=\edge(b)\cap O''$, $\edge(a)\cap O''=\edge(a_1)\cap O''$, and $(\partial I^-(a)\backslash a)\cap O''=(\partial I^-(a_1)\backslash a_1)\cap O''$. Thus, by entropy normal definition and that of $b$ and $O''\sset O'$, we conclude $\edge(b')\cap O''\sset (\edge(a_1)\cap O'')\cup((\partial I^-(a_1)\backslash a_1)\cap O)$. Within $\cM\backslash \cl(O)$, $b\backslash\cl(O)=a\backslash\cl(O)$ implies $b'\backslash\cl(O)=a_1\backslash\cl(O)$, from which it similarly follows that $\edge(b')\backslash\cl(O)=\edge(a_1)\backslash\cl(O)$.  Together this shows that $\edge(b')\sset \edge(a_1)\cup((\partial I^-(a_1)\backslash a_1)\cap O)$, concluding the verification that $b'$ is a past directed outward null deformation of $a_1$ within $O$. 
    
    The entropy normal criteria therefore imply $\Sgen(b')-\Sgen(a_1)\leq 0$. Applying strong subadditivity of $\Sgen$ with auxiliary system $b\backslash b'=a\backslash a_1 =a_2\backslash a_1$ (see above) yields $0\geq \Sgen(b')-\Sgen(a_1)\geq \Sgen(b)-\Sgen(a)$.

\end{proof}

\begin{proof}[Construction of $U_0$:]
    Within $\cM'$, let $\Phi_a$ be the mapping to $\Sig_a$ according to the integral flow of a time function. For any $x\in \edge(a)$, let $\g(x;\lambda)$ be the null geodesic fired to the future and outwards from $a$ starting at $x$, where $\lambda\geq 0$ is the affine parameter and $\g(x;0)=x$. The acausal embedded submanifold 
$(\Sig_a,\text{dist})$
is a metric space with $\text{dist}(z,y)$ being the infimum of the lengths of (piecewise $C^1$) paths in $\Sig_a$ between $z$ and $y$ as measured by $g$. $\Sig_a$ being closed, we have $\cl(a)=\cl_{\Sig_a}(a)=a\cup\edge(a)$, where the second equality follows from Prop. (\ref{prop:edge}) and the subsequent remark. Denote the distance between any $q\in \Sig_a$ and $\cl(a)$ by $\text{dist}(q):=\text{dist}(q,\cl(a))=\inf\{\text{dist}(q,q'):q'\in\cl(a)\}$, and $f(x;\lambda)=\text{dist}(\Phi_a(\g(x;\lambda)))$. 

By continuous dependence of geodesics on initial data, every $x\in \edge(a)$ has a neighborhood $U_x^1$ and an $\lambda_x^1>0$ such that $\g(\cdot;\cdot)$ is defined and continuous over $S_x^1=(\edge(a)\cap U_x^1)\times [0,\lambda_x^1]$. By definition of firing outwards, $\Phi_a(\g(x;\lambda))\in a'$ for  small nonzero $\lambda$, so that $f(x;\lambda_x^2)=f_x$ for some $\lambda_x^2\in(0,\lambda_x^1)$ and $f_x>0$. The composition $f(\cdot;\cdot)$ is also continuous over $S_x^1$, so there is a sub-neighborhood $U_{x}^{2}\sset U_x^1$ of $x$ such that $f(y,\lambda_x^2)\geq f_x/2$ for all $y\in \edge(a)\cap U_x^2$. Then $\{U_x^2\}_{x\in\edge(a)}$ is an open cover of $\edge(a)$, so there is a finite subcover indexed by $x_i$, $i=1,...,m$. Set $l=\min_i(f_{x_i})/3$.

Now, $U_0'=\cM'\cap O_*\cap d^{-1}([0,l))$ is a neighborhood of $\edge(a)$. For each $x\in \edge(a)$, let $U_x^3$ be a neighborhood of $x$ with $\cl(U_x^3)\sset U_0'$ (one can see that this always exists: choose coordinates around $x$, find an open ball contained in $U_0'$, and take $U_x^3$ to be a concentric ball of half the radius). Then $\{U_x^3\}_{x\in\edge(a)}$ is an open cover of $\edge(a)$,  so there is a finite subcover indexed by $x'_i$, $i=1,...,m'$. We set $U_0=\bigcup_{i=0}^{m'}U^3_{x'_i}$. This obeys (1), as $\cl(U_0)=\bigcup_{i=0}^{m'}\cl(U^3_{x'_i})\sset U_0'$. By construction, for each $x\in \edge(a)$, $f(x,\lambda_{x_i})\geq f_{x_i}/2>l$ for some $i=1,...,m$, so $\Phi_a(\g(x,\lambda_{x_i}))\in a'\backslash U_0$. Thus $U_0$ satisfies (2).
\end{proof}

\bibliography{all}

\end{document}